\documentclass[journal,12pt,draftclsnofoot,onecolumn]{IEEEtran}
\usepackage{mathrsfs}
\usepackage[cmex10]{amsmath}
\usepackage{amsfonts,amssymb}
\usepackage{mdwmath}
\usepackage{amsthm}
\usepackage{subfigure}
\usepackage{epsf}
\usepackage{graphicx}
\usepackage{subfigure}
\usepackage[sort,compress]{cite}
\usepackage{algorithmic}
\usepackage[ruled,vlined]{algorithm2e}
\usepackage{booktabs} % To thicken table lines
\usepackage{multirow}
\usepackage{color}
\newtheorem{lemma}{Lemma}

\theoremstyle{definition}

\begin{document}
\begin{titlepage}
\begin{center}
\vspace*{-2\baselineskip}
\begin{minipage}[l]{7cm}
\flushleft
\includegraphics[width=2 in]{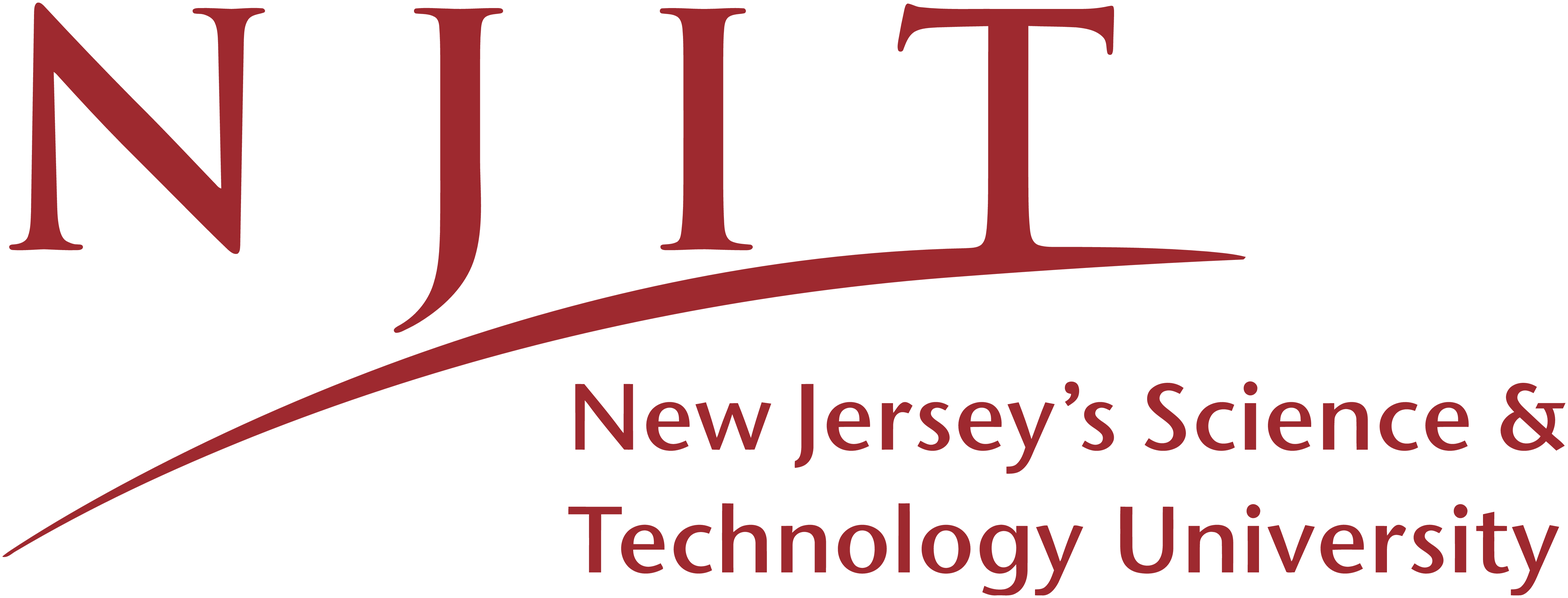}
\end{minipage}
\hfill
\begin{minipage}[r]{7cm}
\flushright
\includegraphics[width=1 in]{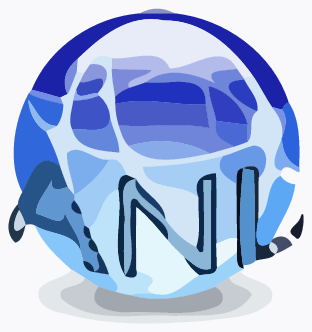}
\end{minipage}
\vfill
\textsc{\LARGE Renewable Energy-Aware Inter-datacenter Virtual Machine \\[12pt]
Migration over Elastic Optical Networks}\\
\vfill
\textsc{\LARGE LIANG ZHANG\\[12pt]
\LARGE TAO HAN \\[12pt]
\LARGE NIRWAN ANSARI}\\
\vfill
\textsc{\LARGE TR-ANL-2015-005\\[12pt]
\LARGE AUGUST 26, 2015}\\[1.5cm]
% Bottom of the page
\vfill
{ADVANCED NETWORKING LABORATORY\\
 DEPARTMENT OF ELECTRICAL AND COMPUTER ENGINEERING\\
 NEW JERSY INSTITUTE OF TECHNOLOGY}
\end{center}
\end{titlepage}
\title{Renewable Energy-Aware Inter-datacenter Virtual Machine Migration over Elastic Optical Networks}
%\author{\IEEEauthorblockN{Liang Zhang, Tao Han, and Nirwan Ansari}
%\IEEEauthorblockA{Advanced Networking Lab., Department of Electrical and Computer Engineering\\
%New Jersey Institute of Technology, New Jersey 07102, USA \\
%Email: \{lz284, th36, nirwan.ansari\}@njit.edu}}
%\author{Liang Zhang, and Nirwan Ansari\\
%%Advanced Networking Lab.\\
%Department of Electrical and Computer Engineering\\
%New Jersey Institute of Technology\\ New Jersey 07102, USA \\
%Email: \{lz284, nirwan.ansari\}@njit.edu \\
%%Swiss Federal Institute of Technology\\ Microcomputing Laboratory \\ IN-F
%%Ecublens, 1015 Lausanne, Switzerland\\ Paolo.Ienne@di.epfl.ch\\
%% For a paper whose authors are all at the same institution,
%% omit the following lines up until the closing ``}''.
%% Additional authors and addresses can be added with ``\and'',
%% just like the second author.
%\and
%Tao Han\\
%Department of Electrical and Computer Engineering\\University of North Carolina at Charlotte\\
%North Carolina 28223, USA\\
%Email: than3@uncc.edu\\}
\author{\IEEEauthorblockN{Liang Zhang, Tao Han, ~\IEEEmembership{Member,~IEEE}, and Nirwan Ansari, ~\IEEEmembership{Fellow,~IEEE}}
\thanks{Liang Zhang and Nirwan Ansari are with the Advanced Networking Laboratory, Department of Electrical and Computer Engineering, New Jersey Institute of Technology, Newark, NJ, 07102, USA. (Email: \{lz284, nirwan.ansari\}@njit.edu).}
\thanks{Tao Han was with the Advanced Networking Laboratory, Department of Electrical and Computer Engineering, New Jersey Institute of Technology. He is now with the Department of Electrical and Computer Engineering, University of North Carolina at Charlotte, Charlotte, NC, 28223, USA. (Email: than3@uncc.edu).}}
\maketitle
\begin{abstract}
Datacenters (\emph{DC}s) are deployed in a large scale to support the ever increasing demand for data processing to support various applications. The energy consumption of DCs becomes a critical issue. Powering DCs with renewable energy can effectively reduce the brown energy consumption and thus alleviates the energy consumption problem. Owing to geographical deployments of DCs, the renewable energy generation and the data processing demands usually vary in different DCs. Migrating virtual machines (\emph{VM}s) among DCs according to the availability of renewable energy helps match the energy demands and the renewable energy generation in DCs, and thus maximizes the utilization of renewable energy. Since migrating VMs incurs additional traffic in the network, the VM migration is constrained by the network capacity. The inter-datacenter (\emph{inter-DC}) VM migration with network capacity constraints is an NP-hard problem. In this paper, we propose two heuristic algorithms that approximate the optimal VM migration solution. Through extensive simulations, we show that the proposed algorithms, by migrating VM among DCs, can reduce up to $31\%$ of brown energy consumption.
\end{abstract}

\begin{IEEEkeywords}
Manycast, Cloud Computing, Elastic Optical Networks.
\end{IEEEkeywords}

\section{Introduction}
Cloud infrastructures are widely deployed to support various emerging applications such as: Google App Engine, Microsoft Window Live Service, IBM Blue Cloud, and Apple Mobile Me \cite{Sadiku2014}. Large-scale data centers (\emph{DC}s), which are the fundamental engines for data processing, are the essential elements in cloud computing \cite{Zhangyan_13DC, zhangyan_dc_survey}. Information and Communication Technology (\emph{ICT}) is estimated to be responsible for about $14\%$ of the worldwide energy consumption by 2020 \cite{Pickavet2008}. The energy consumption of DCs accounts for nearly 25\% of the total ICT energy consumption \cite{Pickavet2008}. Hence, the energy consumption of DCs becomes an imperative problem.

Renewable energy, which includes solar energy and wind power, produces $12.7\%$ domestic electricity of the United States in 2011 \cite{Green-cloud2013}. Renewable energy will be widely adopted to reduce the brown energy consumption of ICT \cite{tao2014_magazine}. For example, Parasol is a solar-powered DC \cite{Goiri_GreenDc}. In Parasol, GreenSwitch, a management system, is designed to manage the work loads and the power supplies \cite{Goiri_GreenDc}. The availability of renewable energy varies in different areas and changes over time. The work loads of DCs also vary in different areas and at different time. As a result, the renewable energy availability and energy demands in DCs usually mismatch with each other. This mismatch leads to inefficient renewable energy usage in DCs. To solve this problem, it is desirable to balance the work loads among DCs according to their green energy availability. Although the current cloud computing solutions such as cloud bursting \cite{Wood_cloudNet_2014}, VMware and F5 \cite{VMware_F5} support inter-datacenter (\emph{inter-DC}) virtual machine (\emph{VM}) migration, it is not clear how to migrate VM among renewable energy powered DCs to minimize their brown energy consumption.

Elastic Optical Networks (\emph{EONs}), by employing orthogonal frequency division multiplexing (\emph{OFDM}) techniques, not only provide a high network capacity but also enhance the spectrum efficiency because of the low spectrum granularity \cite{Shieh2007}. The granularity in EONs can be 12.5 GHz or even much smaller \cite{Armstrong2009}. Therefore, EONs are one of the promising networking technologies for inter-DC networks \cite{Develder2012}.

Powering DCs with renewable energy can effectively reduce the brown energy consumption, and thus alleviate green house gas emissions. DCs are usually co-located with the renewable energy generation facilities such as solar and wind farms \cite{Figuerola_2009}. Since transmitting renewable energy via the power grid may introduce a significant power loss, it is desirable to maximize the utilization of renewable energy in the DC rather than transmitting the energy back to the power grid.  In this paper, we investigate the \emph{r}enewable \emph{e}nergy-\emph{a}ware \emph{i}nter-DC VM \emph{m}igration (\emph{RE-AIM}) problem that optimizes the renewable energy utilization by migrating VMs among DCs. Fig. \ref{fig:six-cloud} shows the architecture of an inter-DC network. The vertices in the graph stand for the optical switches in EONs. DCs are connected to the optical switches via IP routers \footnote{In this paper, we focus on the EONs. The design and optimization of the IP networks are beyond the scope of this paper.}. These DCs are powered by hybrid energy including brown energy, solar energy, and wind energy. In migrating VMs among DCs, the background traffic from other applications are also considered in the network. For example, assume that DC $1$ lacks renewable energy while DC $2$ and DC $3$ have superfluous renewable energy. Some VMs can be migrated out of DC $1$ in order to save brown energy. Because of the background traffic and the limited network resource, migrating VMs using different paths (Path $1$ or Path $2$) has different impacts on the network in terms of the probability of congesting the network. It is desirable to select a migration path with minimal impact on the network.
%In the previous work \cite{zhang-osa}, the anycast routing problem is well investigated in inter-DC networks. Base on previous work, this paper target to exploit the potential renewable energy between diversified DCs with balancing the unequally distributed renewable energy through live VM migration.

\begin{figure}[!htb]
    \centering
    \includegraphics[width=0.95\columnwidth]{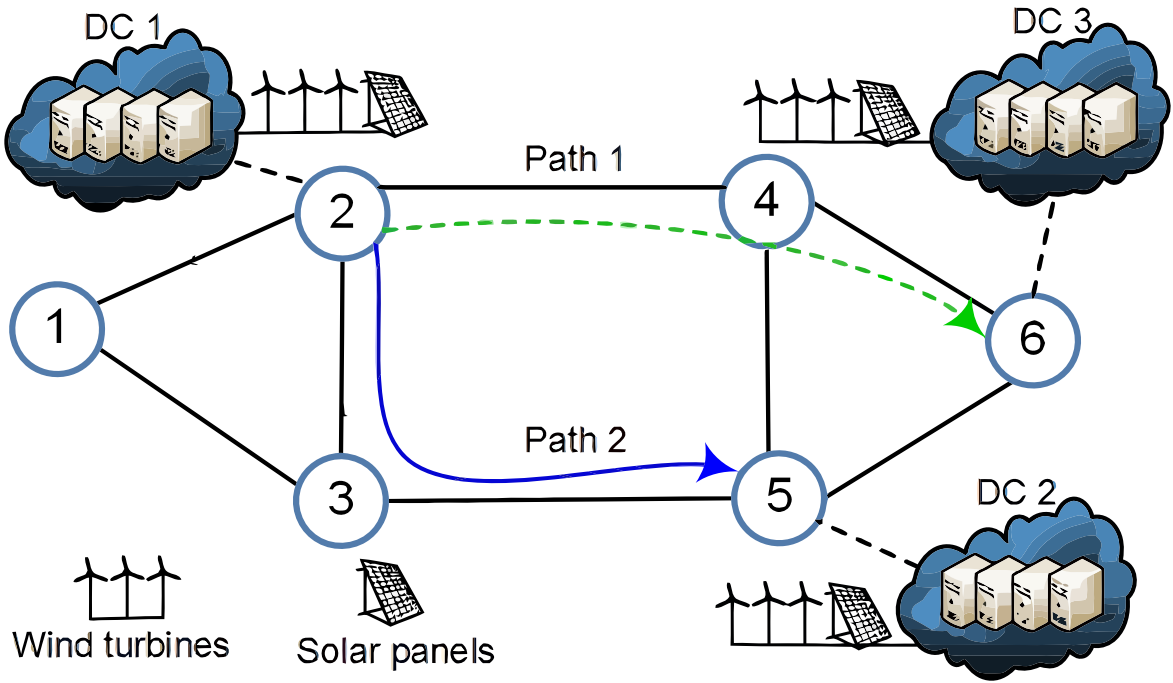}
    \caption{\small Inter-DC architecture.}
    \label{fig:six-cloud}
\end{figure}

The rest of this paper is organized as follows. Section \ref{sec:related_work} describes the related work. Section \ref{sec:problem} formulates the RE-AIM problem. Section \ref{sec:analysis-algorithms} briefly analyzes the property of the RE-AIM problem and proposes two heuristic algorithms to solve the problem. Section \ref{sec:evaluations} demonstrates the viability of the proposed algorithms via extensive simulation results. Section \ref{conclusion} concludes the paper.

\section{Related Work}\label{sec:related_work}
Owing to the energy demands of DCs, many techniques and algorithms have been proposed to minimize the energy consumption of DCs \cite{Ghamkhari_2013}.

Fang \emph {et al.} \cite{Yonggang_Wen_Globemcom_2012} presented a novel power management strategy for the DCs, and their target was to minimize the energy consumption of switches in a DC. Cavdar and Alagoz \cite{Survey_Green_DC_2012} surveyed the energy consumption of server and network devices of intra-DC networks, and showed that both computing resources and network elements should be designed with energy proportionality. In other words, it is better if the computing and networking devices can be designed with multiple sleeping states. A few green metrics are also provided by this survey, such as Power Usage Effectiveness (PUE) and Carbon Usage Effectiveness (CUE).

Deng \emph {et al.} \cite{Fangming_Liu_IEEE_Network2014} presented five aspects of applying renewable energy in the DCs: the renewable energy generation model, the renewable energy prediction model, the planning of green DCs (i.e., various renewable options, avalabity of energy sources, different energy storage devices), the intra-DC work loads scheduling, and the inter-DC load balancing. They also discussed the research challenges of powering DCs with renewable energy. Ghamkhari and Mohsenian-Rad \cite{Ghamkhari_2013} developed a mathematical model to capture the trade-off between the energy consumption of a data center and its revenue of offering Internet services. They proposed an algorithm to maximize the revenue of a DC by adapting the number of active servers according to the traffic profile. Gattulli \emph {et al.} \cite{IP_over_WDM_icc_2012} proposed algorithms to reduce $CO_{2}$ emissions in DCs by balancing the loads according to the renewable energy generation. These algorithms optimize renewable energy utilization while maintaining a relatively low blocking probability.

%\textcolor{blue} {Fiorani \emph {et al.} \cite{Intra_DC_with_Core_Edge_2014} presented an architecture of intra-DC and core network based on hybrid optical switching, they also proposed an architecture of hybrid optical switching edge node; the advantage of the joint core and intra-DC network architecture is that it can reduce the inefficient energy electronic interfaces between DCs and core networks .}

%Wang \emph {et al.} \cite{Mengqing_2015} studied green energy aware intra-DC VM migration problem for renewable energy powered DCs. The objective is to minimize the total energy consumption of servers and cooling system. They proposed a joint optimal planing strategy, while this strategy includes the solar energy prediction and the genetic algorithm.

Mandal \emph {et al.} \cite{Green-cloud2013} studied green energy aware VM migration techniques to reduce the energy consumption of DCs. They proposed an algorithm to enhance the green energy utilization by migrating VMs according to the available green energy in DCs. However, they did not consider the network constraints while migrating VMs among DCs. In the optical networks, the available spectrum is limited. The large amount of traffic generated by the VM migration may congest the optical networks and increase the blocking rate of the network. Therefore, it is important to consider the network constraints in migrating VMs. In this paper, we propose algorithms to solve the green energy aware inter-DC VM migration problem with network constraints.

\section{Problem Formulation}
\label{sec:problem}
In this section, we present the network model, the energy model, and the formulation of the RE-AIM problem. The key notations are summarized in Table \ref{tab:notations}.

\begin{table}[!htb]
\begin{center}
\caption{The Important Notations}\label{tab:notations}
\begin{tabular}{{|l|p{185pt}|}}
\hline
Symbol           & Definiton                                    \\
\hline
$c_{e}$          & The capacity of a link $e \in E$ in terms of spectrum slots.\\
$c_{s}$          & The capacity of a spectrum slot.    \\
$c_{m}$          & The maximum number of servers in the $m$th DC. \\
$\varsigma$      & The maximum number of VMs can be supported in a server. \\
$\varPhi_{m}$    & The amount of renewable energy in the $m$th DC. \\
$\varTheta_{m}$  & The number of VMs in the $m$th DC.\\
$\alpha_{m}$     & Per unit energy cost for the $m$th DC.        \\
$\zeta_{m,k}$    & The required bandwidth for migrating the $k$th VM in the $m$th DC. \\
$\mathcal {R}$   & The set of the migration requests.\\
$\mathcal {Q}_{r}$& The set of VMs migrated in the $r$th migration.\\
$\kappa$         & The migration granularity.\\
$w_{p}^{r,m}$    & The used spectrum slot ratio of the $p$th path in the $r$th migration from the $m$th DC.\\
$w_{B}$          & The maximum network congestion ratio.\\
$p_{s}$          & The maximum energy consumption of a server.      \\
$\eta$           & The power usage efficiency.                                  \\
\hline
\end{tabular}
\end{center}
\end{table}

\subsection{Network Model}\label{sec:network-model}
We model the inter-DC network by a graph, $\mathcal{G}(V, E, B)$. Here, $V$, $E$ and $B$ are the node set, the link set and the spectrum slot set, respectively. The set of DC nodes is denoted as $\mathcal{D}$. We assume that all DCs are powered by hybrid energy. We denote $\mathcal {D}_{s}$ as the set of DCs that does not have sufficient renewable energy to support their work loads and $\mathcal {D}_{d}$ as the set of DCs that has surplus renewable energy. During the migration, $\mathcal {D}_{s}$ and $\mathcal {D}_{d}$ correspond to the two sets of DCs acting as the sources and destinations, respectively. We define $\kappa$ as the migration granularity, which determines the maximum routing resource that can be used in one migration to each DC.

\subsection{Energy Model}\label{sec:energy-model}
We assume that there are $c_{m}$ servers in the $m$th DC and each server can support up to $\varsigma$ VMs.
The energy consumption of a server is $p_{s}$ when it is active. A server is active as long as it hosts
at least one active VM; otherwise, the server is in the idle state. Here, we assume that an idling server will be turned off and its energy consumption is zero. Then, $\left\lceil \varTheta_{m}/\varsigma \right\rceil$ is the number of active servers required in the $m$th DC \cite{Green-cloud2013}. We denote $\eta$ as the power usage effectiveness, which is defined as the ratio of a DC's total energy consumption (which includes the facility energy consumption for cooling, lighting, etc. \cite{abbas2015_green_dc}) to that of the servers in the DC. Given $\eta$, a DC's total energy consumption is $\eta \cdot p_{s} \cdot \varTheta_{m}/\varsigma$. We denote $\varUpsilon_{m}$ as the brown energy consumption in the $m$th DC. Then,
\begin{equation}\label{eq:v1}
\varUpsilon_{m}=\max(0, \eta \cdot p_{s} \cdot \left\lceil \varTheta_{m}/\varsigma\right\rceil-\varPhi_{m})
\end{equation}

\subsection{Problem Formulation}\label{sec:formulation}
In the problem formulation, $\chi_{p,f}^{m,k}$ is a binary variable. $\chi_{p,f}^{m,k}=1$ indicates that the $k$th VM in the $m$th DC is migrated using the $p$ path with the $f$th spectrum slot as the starting spectrum slot. The objective of the RE-AIM problem is to minimize the total brown energy cost in all DCs with the VM service constraints and the network resource constraints. The problem is formulated as:
\begin{align}
\label{eq:objective}
\min_{\chi_{p,f}^{m,k}}\quad & \sum_{m} {\alpha_{m} \cdot \varUpsilon_{m}}\\
s.t.:\nonumber & \\
& VM\;service\;constraints:\nonumber\\
&\sum_{m}\sum_{k}\sum_{p}\sum_{f}\chi_{p,f}^{m,k}=\sum_{m}\varTheta_{m}\label{eq:c1}\\
&\sum_{k}\sum_{p}\sum_{f}\chi_{p,f}^{m,k} \leq c_{m}, \forall m\in \mathcal {D}_{s}\label{eq:c2}\\
&\begin{aligned}\label{eq:c3}
  &\sum_{m'\in \mathcal {D}_{s}}\sum_{k}\sum_{p}\sum_{f}\chi_{p,f}^{m',k} + \\
  &\sum_{k}\sum_{p}\sum_{f}\chi_{p,f}^{m,k} \leq c_{m},\forall m\in \mathcal {D}_{d}
  \end{aligned}\\
& Network\;resource\;constraints:\nonumber\\
& w_{p}^{r,m}+ \frac{\varGamma_{p,f}^{r,m}}{c_{e}}\leq w_{B},\quad\forall m\in \mathcal {D}_{s}, r \in\mathcal{R} \label{eq:c4}\\
& f(\chi_{p,f}^{m,k})+b(\chi_{p,f}^{m,k}) \leq c_{e}\label{eq:c5} \\
& f(\chi_{p,f}^{m,k})+b(\chi_{p,f}^{m,k})-f(\chi_{p,f}^{m,k+1})\leq 0\label{eq:c6}\\
& \begin{aligned}\label{eq:c9}
& f(i)+b(i)-f(j)\leq [2-\delta_{i,j}-y(i, j)] \cdot\\
& F_{max}, \quad\forall i \neq j
  \end{aligned}\\
& \begin{aligned}\label{eq:c10}
& f(j)+b(j)-f(i)\leq [1+\delta_{i,j}-y(i, j)] \cdot\\
&F_{max}, \quad\forall i \neq j
  \end{aligned}
\end{align}
%&f(j)-f(i) \leq F_{max} \cdot \delta_{i,j}, \quad\forall i \neq j \label{eq:c7}\\
%&f(i)-f(j) \leq F_{max} \cdot (1-\delta_{i,j}), \quad\forall i \neq j \label{eq:c8}\\
Here, Eqs. \eqref{eq:c1}-\eqref{eq:c3} are the VM service constraints. Eq. \eqref{eq:c1} constrains that all the VMs should be hosted in the DCs, while Eqs. \eqref{eq:c2}-\eqref{eq:c3} constrain that the total number of VMs in a DC should not exceed the DCs' capacity. The network resource constraints are shown in Eqs. \eqref{eq:c4}-\eqref{eq:c10}.

Eq.\eqref{eq:c4} constrains the network congestion ratio to be less than $w_{B}$, which is the maximum network congestion ratio allowed for routing in the network. In Eq. \eqref{eq:c4}, $w_{p}^{r,m}$ is the spectrum slot ratio of the $p$th path in the $r$th migration from the $m$th DC, which is defined as the ratio of the number of occupied spectrum slots in the $p$th path to the total number of spectrum slots of this path. $\varGamma_{p,f}^{r, m}$ is defined as the number of spectrum slots used in the $p$th path for the $r$th migration from the $m$th DC. %Then,
%$w_{p}^{r,m}$ the used spectrum slot ratio of the $m$th DC in the $r$th migration, which is the number of used spectrum slots of the $m$th DC in the $r$th migration, over the number of available spectrum slots of the $p$th path in the $r$th migration from the $m$th DC.
%%%%%%%%%%%%%%%%%%%%%------ to be done ---------------%%%%%%%%%%%%%%%%%%
%\begin{equation}\label{eq:c4-2}
%\varGamma_{p,f}^{r,m}= \min(\kappa, b(\mathcal {Q}_{r,m}))
%%\varGamma_{p,f}^{r}= \min(\kappa, b(\chi_{p,f}^{m,k})+f(\chi_{p,f}^{m,k})-f(\chi_{p,f}^{m,k+1}))
%\end{equation}
%Here, $\Mathcal {Q}_{R,M}$ Is The Vm Request Set Of The $R$Th Migration From The $M$Th Dc, $\Mathcal {Q}_{R,M}$ Is A %Subset Of $\Mathcal {Q}_{R}$, $\Mathcal {Q}_{R}=\{\Mathcal {Q}_{R,M}, \Forall M\}$, And $B(\Cdot)$ Is The Bandwidth %Requirement In Terms Of Spectrum Slots. The Granularity Constraint, Also Included In Eq. \Eqref{Eq:C4-2}, Targets To Make %The Migration More Reliable, And It Is Also A Leverage Between The Complexity And Energy Saving.
Eq.\eqref{eq:c5} is a link capacity constraint of the network; it constrains the bandwidth used in migrating VMs not to exceed the capacity of the network resource. Here, $b(\cdot)$ is the bandwidth requirement in terms of spectrum slots, and $f(\cdot)$ is the index of the starting spectrum slot of a path. For example, $f(\chi_{p,f}^{m,k})$ represents the starting spectrum slot index of the path, which is used by $\chi_{p,f}^{m,k}$. Eq.\eqref{eq:c6} is the spectrum non-overlapping constraint of a path used by two different VMs in one migration. This constraint must be met for each VM in every migration; if two VMs use the same spectrum slot in one migration, the total bandwidth allocated to the two VMs should not exceed the capacity of a spectrum slot; otherwise, each VM must use a unique spectrum slot. In the migration, the VMs are sorted in ascending order based on their bandwidth requirement. We assume the VMs are migrated according to an ascending order; for example, the $(k+1)$th VM is moved after the $k$th VM is migrated.

Eqs. \eqref{eq:c9}-\eqref{eq:c10} are the spectrum non-overlapping and the continuity constraints \cite{Christodoulopoulos2011}. This spectrum non-overlapping constraint is used for different paths. In these constraints, $i$ and $j$ represent two different paths used in the migration. Here, $F_{max}$ is the upper bound of the total bandwidth requirement in terms of spectrum slots. $\delta_{i,j}$ $(\forall i \neq j)$ is a Boolean variable defined in Eq. \eqref{eq:c15}, which equals $1$ if the starting spectrum slot index of the $i$th path is smaller than that of the $j$th path; otherwise, it is $0$. We define $y(i,j)$ $(\forall i \neq j)$ as a Boolean indicator, which equals $1$ if the $i$th path and the $j$th path in the migration have at least one common link; otherwise, it is $0$. We give an example to illustrate these equations. If $y(i, j)=1$ and $\delta_{i,j}=1$, Eq. \eqref{eq:c9} becomes Eq. \eqref{eq:c16}, which ensures the bandwidth non-overlapping constraint. Eq. \eqref{eq:c10} is automatically satisfied in this case.

%For example, if $f(i)<f(j)$, then Eq. \eqref{eq:c7} ensures $\delta_{i,j}=1$, Eq. \eqref{eq:c8} is relaxed. %If $f(i)>f(j)$, Eq. \eqref{eq:c8} ensures $\delta_{i,j}=0$, Eq. \eqref{eq:c7} is relaxed. Eqs. %\eqref{eq:c9}-\eqref{eq:c10} exclude $f(i)=f(j)$, if the $i$th path and the $j$th path in the migration have %one or more common links.
%\begin{equation}\label{eq:c13}
%F_{max}=\sum_{r\in R} {b(\mathcal {Q}_{r})}
%\end{equation}
\begin{equation}\label{eq:c15}
\delta_{i,j}=
\begin{cases}
1, & f(i) < f(j)\\
0, & f(i) \geq f(j)
\end{cases}
\end{equation}
%\delta_{i,j}+ \delta_{j,i}=1, \forall i \neq j.
\begin{equation}\label{eq:c16}
  f(i)+b(i)\leq f(j)
\end{equation}

When we provision spectrum slots for requests in the EONs, the path continuity constraint, spectrum continuity constraint and non-overlapping constraint must be considered. For the path continuity constraint, a lightpath must use the same subcarriers in the whole path for a request. For the spectrum continuity constraint, the used subcarriers must be continuous if a request needs more than one subcarriers. For the non-overlapping constraint, two different lightpaths must be assigned with different subcarriers if they have one or more common links. Since we use a path based method to formulate the RE-AIM problem, the path continuity constraint of the network is already taken into account.

The main contribution of this paper is considering the network influence on the migration when we minimize the brown energy consumption of the DCs. In other words, we want to impose a controllable effect on the network in the migration that leads to less network congestion.

% \begin{equation}\label{eq:c9}
%  \sum_{p}\sum_{f}(\varGamma_{p,f}-L_{p}^{l}) \leq c_{l},\quad\forall l \in L
% \end{equation}
% \begin{equation}\label{eq:c10}
%  \sum_{p}(\varGamma_{p,f} \cdot \sum_{f}L_{p}^{l}) \leq c_{l},\quad\forall l \in L
%  \end{equation}
%
%Eq.\eqref{eq:c10} is used to explain Eq.\eqref{eq:c9}. The original network link capacity constraint should be expressed as Eq.\eqref{eq:c10}, we reverse the definition of $L_{i}^{l}$ in order to make it linearity, and then Eq.\eqref{eq:c10} is converted into Eq.\eqref{eq:c9}. If the definition of current $L_{i}^{l}$ is reversed, Eq.\eqref{eq:c10} guarantees that the used spectrum slots should not exceed the slot capacity of each link.

\section{Problem Analysis and Heuristic Algorithms }\label{sec:analysis-algorithms}
\subsection{Problem Analysis}\label{sec:analysis}
To solve the RE-AIM problem, both the energy costs in DCs and the network resource required for the migration should be considered. For example, when a DC consumes brown energy, it is desirable to migrate some VMs to other DCs. The VM migration will introduce additional traffic to the network. To avoid congesting the network, we have to optimize the number of VMs that will be migrated and select the routing path for the migration. Therefore, it is challenging to solve the RE-AIM, which is proven to be NP-hard.

\begin{lemma}
The RE-AIM problem is NP-hard
\end{lemma}
\begin{proof}
We prove that the RE-AIM problem is NP-hard by reducing any instance of the multi-processor scheduling problem (\emph{MPS}) into the RE-AIM problem.\end{proof}
%Since the MPS is an NP-hard problem \cite{Christodoulopoulos2011}, the RE-AIM problem is also NP-hard. For the sake of brevity, we omit the detailed proof. \end{proof}

In the RE-AIM problem, without considering the network constraints, the optimal number of VMs hosted in the DCs can be derived according to the availability of the renewable energy. However, with the consideration of the network constraints and the background traffic, it is difficult and impossible to solve the RE-AIM problem online. For the RE-AIM problem, many VMs are migrated from a set of DCs (source DCs) to another set of DCs (destination DCs). Therefore, we can model the VM migration problem as a manycast problem. Since the RE-AIM problem is NP-hard, we propose heuristic algorithms to solve this problem. These algorithms determine which VM should be migrated to which DC and select a proper routing path in the network to avoid congesting the network. We consider two network scenarios. The first one is a network with light traffic load. Under this network scenario, we design Manycast with Shortest Path Routing (\emph{Manycast-SPR}) algorithm for VM migrations. The second network scenario is a network with heavy traffic load. In this case, we propose Manycast Least-Weight Path Routing (\emph{Manycast-LPR}) for migrating VMs among DCs.

\subsection{Heuristic Algorithms For Light Work Loads}\label{sec:algorithm_Manycast-SPR}
When the network load is light, there are more available spectrum slots. It is easy to find a path with available spectrum slots for the migration requests. Then, a lower computing complexity algorithm is preferred. Manycast-SPR only uses the shortest path, and thus it is a very simple algorithm. Hence, Manycast-SPR is expected to provision the inter-DC VM migration requests in a network with light work loads.

The Manycast-SPR algorithm, as shown in Alg. \ref{Manycast-SPR}, is to find the shortest routing path that satisfies the VM migration requirement and the network resource constraints. In the beginning, we input $\mathcal{G}(V, E, B)$, $\varTheta_{m}$ and $\varPhi_{m}$, and then calculate the optimal work loads distribution. Afterward, we get $\mathcal {D}$, $\mathcal {D}_{s}$ and $\mathcal {D}_{d}$. Then, we collect the migration requests $R$. Here, our algorithm splits the manycast requests into many anycast requests $r \in R$. Now, we start to find a source DC $s$ and a destination DC $d$ for the request $r$. The migration will try to use the shortest path $p$ from $s$ to $d$; the request $r$ is carried out if the network congestion constraint is satisfied; otherwise, the request is denied. Then, we update $\mathcal {D}_{s}$ and $\mathcal {D}_{d}$ for the next request. After many rounds of migration, if $\mathcal {D}_{s}$ or $\mathcal {D}_{d}$ is empty, or Eq. (\ref{eq:c4}) is not satisfied, the migration is completed.

Details of the Manycast-SPR algorithm is described in \emph{Algorithm} \ref{Manycast-SPR}. Here, $p(\cdot)$ is a function which targets to get the path for the migration. The complexity of Manycast-SPR is $O(|B| |E|^{2} \mathcal{|R|} \mathcal|{Q}_{r}| \mathcal {|D|}^{2}c_{m}\varsigma)$. Here, $O(\mathcal {|D|}^{2}c_{m}\varsigma)$ is the complexity to determine the optimal work loads, $O(|B|)$ is the complexity to determine $f$, and $O(\mathcal{|R|}\mathcal|{Q}_{r}|)$ is the complexity in building the VM set for the migration. $O(|E|^{2})$ is the complexity of determining the path $p$ for Manycast-SPR.
%$\mathcal {D}_{s}$, $\mathcal {D}_{d}$ and the migratory traffic are developed by the optimal work loads of the ideal network scenario. A node pair $(s, d)$ $(s \in \mathcal {D}_{s}, d \in \mathcal {D}_{d})$ is built to set up a routing path $p$. After calculating the network influence $\varOmega_{p}^{r}$ of the $p$th path, the VM migratory request $r$ is determined to stay or to fly. Manycast-SPR is a simply method to achieve the VM migration. It expends less time to figure out the migration, and it is suitable for the scenario which cares about the calculating time.

\subsection{Heuristic Algorithms For Heavy Work Loads}\label{sec:algorithm_Manycast-LPR}
When the work load of the network is heavy, the number of available spectrum slots in the network is limited. Since Manycast-SPR only uses the shortest path (one path) for routing, it is impossible for Manycast-SPR to find an available path and spectrum slots in this scenario. Then, Manycast-SPR may block the migration request, and leads to high brown energy consumption of DCs. Hence, we propose another algorithm Manycast-LPR to achieve better routing performance, that results in low brown energy consumption. Manycast-LPR checks $K$-shortest paths from the source node to the destination node, and picks up the idlest path to serve the requests. The requests will be provisioned with a higher probability by Manycast-LPR as compared to Manycast-SPR. In summary, Manycast-LPR is expected to provision the inter-DC VM migration requests under a heavy work load. It targets to find a path with more available spectrum slots at the expense of a higher complexity.

Manycast-LPR, as shown in Alg. \ref{Manycast-LPR}, is to find the least weight routing path that satisfies the VM migration requirement and the network resource constraints. The main difference between Manycast-LPR and Manycast-SPR is using different ways to find a path. For Manycast-SPR, it first determines the source node and the destination node. Manycast-LPR, however, finds the path first, then uses the path to find the source node and the destination node. The other steps are almost the same. Since Manycast-LPR should calculate the weights for all node pairs to find a path, it increases the complexity.
%is designed to receive a better result for the VM migration, it engages the least weight path for routing in the migration. It tends to finds a path, with more available bandwidth resource and less weight, which satisfying the network influence constraint for routing. After built $\mathcal {D}_{s}$, $\mathcal {D}_{d}$, we compute K-shortest paths of all node pairs $(s, d)$ $(s \in \mathcal {D}_{s}, d \in \mathcal {D}_{d})$, then a least weight path $p$ is acquired with a dedicate node pair $(s, d)$. Then the VM migratory request $r$ is determined to stay or to fly. Manycast-LPR is more complicated than Manycast-SPR, it is designed to achieve a better result. Manycast-LPR is especially suitable for time-careless scenario.

Details of the Manycast-LPR algorithms is described in \emph{Algorithm} \ref{Manycast-LPR}. The complexity for Manycast-LPR is $O(K |B| |E|^{2} \mathcal{|R|} \mathcal|{Q}_{r}| \mathcal{|D|}^{3} c_{m}\varsigma)$. Here, $p(\cdot)$ is a function which targets to get the path for the migration.  $O(\mathcal {|D|}^{2}c_{m}\varsigma)$ is the complexity to determine the optimal work loads, $O(|B|)$ is the complexity to determine $f$, and $O(\mathcal{|R|}\mathcal|{Q}_{r}|)$ is the complexity in building the VM set for the migration. $O(K|E|^{2}\mathcal {|D|})$ is the complexity of determining the path $p$ for Manycast-LPR. The most complex part is to determine the set of VMs for the migration.

%Details of the two algorithms are described in \emph{Algorithm} \ref{Manycast-SPR} and \emph{Algorithm} \ref{Manycast-LPR}, respectively. Here, $p(\cdot)$ is a function which targets to get the path for the migration. The complexity for Manycast-SPR and Manycast-LPR are $O(|B| |E|^{2} \mathcal{|R|} \mathcal|{Q}_{r}| \mathcal {|D|}^{2}c_{m}\varsigma)$ and $O(K |B| |E|^{2} \mathcal{|R|} \mathcal|{Q}_{r}| \mathcal{|D|}^{3} c_{m}\varsigma)$, respectively.  $O(\mathcal {|D|}^{2}c_{m}\varsigma)$ is the complexity to determine the optimal work loads, $O(|B|)$ is the complexity to determine $f$, and $O(\mathcal{|R|}\mathcal|{Q}_{r}|)$ is the complexity in building the VM set for the migration. $O(|E|^{2})$ and $O(K|E|^{2}\mathcal {|D|})$ are the complexities of determining the path p for Manycast-SPR and Manycast-LPR, respectively. The most complex part is to determine the set of VMs for migration.
\begin{algorithm}
\SetKwData{Left}{left}\SetKwData{This}{this}\SetKwData{Up}{up}
\SetKwFunction{Union}{Union}\SetKwFunction{FindCompress}{FindCompress}
\SetKwInOut{Input}{Input}\SetKwInOut{Output}{Output}
\Input{$\mathcal{G}(V, E, B)$, $\varTheta_{m}$ and $\varPhi_{m}$\;}
\Output{$\mathcal {D}$, $\mathcal {Q}_{r}$, $p(\mathcal {Q}_{r})$, $f(\mathcal {Q}_{r})$ and $\varGamma_{p,f}^{r}$, $r\in R$\;}
\nl Build $\mathcal {D}$, $\mathcal {D}_{s}$ and $\mathcal {D}_{d}$ by the the optimal work loads allocation\;
\nl Collect manycast requests $R$\;
\nl \While{$\mathcal {D}_{s}$ and $\mathcal {D}_{d}$ are not empty}{
\nl     calculate the network congestion ratio for all $p$, and get $w_{B}$\;
\nl       \For{all nodes $s \in \mathcal {D}_{s}$}{
\nl            find $s$ with the max migratory VMs as the source node\;}
\nl       \For {all nodes $d \in \mathcal {D}_{d}$}{
\nl          find $d$ with the max available renewable energy as the destination node\;}
\nl          get the shortest path $p(\mathcal {Q}_{r})$ from $s$ to $d$ for the $r$th migration\;
\nl          build $\mathcal {Q}_{r}$ for the $r$th migration according to $p(\mathcal {Q}_{r})$ and Eqs. \eqref{eq:c9}-\eqref{eq:c10}\;
\nl    \If  {Eq. (\ref{eq:c4}) is satisfied}{
\nl          path $p(\mathcal {Q}_{r})$ is used to migrate\;
\nl          find the start spectrum slot index $f(\mathcal {Q}_{r})$ in $B$ \;
\nl          get the allocated bandwidth $\varGamma_{p,f}^{r}$ \;
\nl       update $\mathcal {D}_{s}$ and $\mathcal {D}_{d}$\;}
\nl   \Else{
\nl          return\; } }
\caption{Manycast with Shortest Path Routing\label{Manycast-SPR}}
\end{algorithm}

\begin{algorithm}
\SetKwData{Left}{left}\SetKwData{This}{this}\SetKwData{Up}{up}
\SetKwFunction{Union}{Union}\SetKwFunction{FindCompress}{FindCompress}
\SetKwInOut{Input}{Input}\SetKwInOut{Output}{Output}
\Input{$\mathcal{G}(V, E, B)$, $\varTheta_{m}$ and $\varPhi_{m}$\;}
\Output{$\mathcal {D}$, $\mathcal {Q}_{r}$, $p(\mathcal {Q}_{r})$, $f(\mathcal {Q}_{r})$ and $\varGamma_{p,f}^{r}$, $r\in R$\;}
\nl Build $\mathcal {D}$, $\mathcal {D}_{s}$ and $\mathcal {D}_{d}$ by the the optimal work loads allocation\;
\nl Collect manycast requests $R$\;
\nl \While{$\mathcal {D}_{s}$ and $\mathcal {D}_{d}$ are not empty}{
\nl     calculate the network congestion ratio for all $p$, and get $w_{B}$\;
\nl       \For{all nodes $s \in \mathcal {D}_{s}$}{
\nl       \For {all nodes $d \in \mathcal {D}_{d}$}{
\nl           build K-shortest path set $\mathcal {P}$\;
\nl       \For  {path $p \in \mathcal {P}$}{
\nl          get path $p(\mathcal {Q}_{r})$ with the lowest congestion ratio for the $r$th migration\;   }}}
\nl          build $\mathcal {Q}_{r}$  for the $r$th migration according to $p(\mathcal {Q}_{r})$ and Eqs. \eqref{eq:c9}-\eqref{eq:c10}\;
\nl    \If  {Eq. (\ref{eq:c4}) is satisfied}{
\nl          path $p(\mathcal {Q}_{r})$ is used to migrate\;
\nl          find the start spectrum slot index $f(\mathcal {Q}_{r})$ in $B$\;
\nl          get the allocated bandwidth $\varGamma_{p,f}^{r}$\;
\nl       update $\mathcal {D}_{s}$ and $\mathcal {D}_{d}$\;}
\nl   \Else{
\nl          return\; } }
\caption{Manycast with Least-weight Path Routing\label{Manycast-LPR}}
\end{algorithm}

\section{Performance Evaluations}\label{sec:evaluations}
We evaluate the proposed algorithms for the RE-AIM problem in this section. In order to make the RE-AIM problem simple, we assume migratory VMs can be completed in one time slot. The NSFNET topology, shown in Fig. \ref{fig:NSF-green}, is used for the simulation. There are 14 nodes, and the DCs are located at $\mathcal {D}= \{3, 5, 8, 10, 12\}$ \cite{Zhang2014, zhang-osa}. The DCs are assumed to be equipped with wind turbines and solar panels, which provide the DCs with renewable energy, as shown in Fig \ref{fig:NSF-green}. The constant $\alpha$ is randomly generated from $[1.6, 3.2]$ and represents the varying price of the electric grid. The capacity of a spectrum slot $b$ is set to 12.5Gps. The maximum number of slots $c_{e}$ is set to 300; 300 spectrum slots are available when the network is empty. Assume $\varsigma$ equals to 10; 10 VMs can be run in one server. $K$ is set to 3, i.e., the maximum number of shortest paths that can be used in Manycast-LPR is 3. Without losing generality, the average energy consumption of a VM is assumed to be 1 unit, implying that $p_{s}$ equals to 10 units. The VM bandwidth requirement $\zeta_{m,k}$ is randomly selected from $[1, 14]$, which is convenient for the simulation. The migration requests are generated by the optimal work loads distribution which is calculated based on $\varTheta_{m}$ and $\varPhi_{m}$. The background traffic is randomly generated between node pairs in the network. The background traffic load is counted as an average of $\frac{\lambda}{\mu}$, where $\lambda$ is an average arrival rate of the requests and $\frac{1}{\mu}$ is the holding period of each request \cite{Zhang2014}. Here, the background traffic arriving process is a poisson process, and the holding time is a negative exponential distribution. Parameters which are used for the evaluation are summarized in Table \ref{tab:simulation-parameters}.

\begin{figure}[t]
    \centering
    \includegraphics[width=0.95\columnwidth]{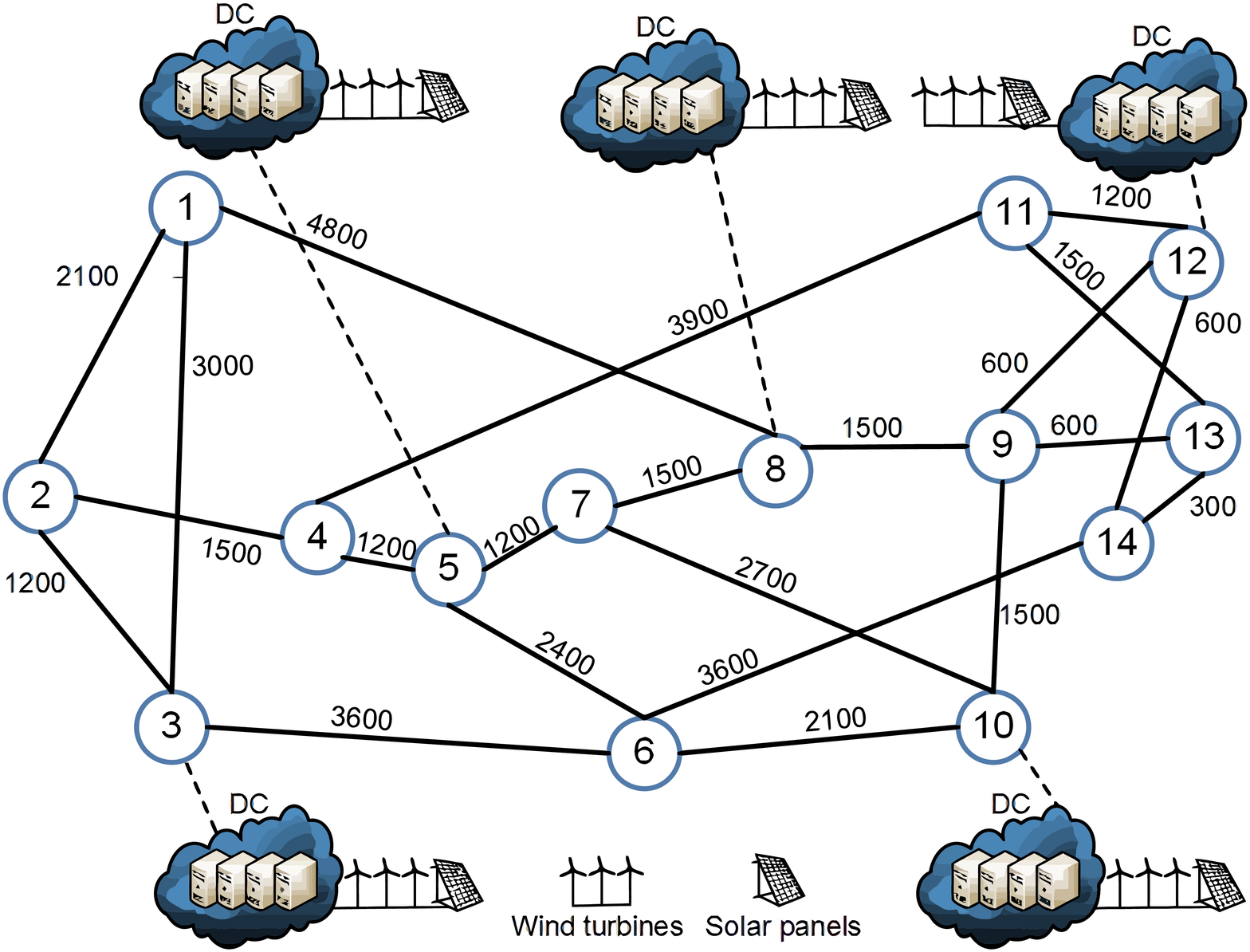}
    \caption{\small NSFNET topology with renewable energy DCs.}
    \label{fig:NSF-green}
\end{figure}

\begin{table}[!htb]
\begin{center}
\caption{\small Simulation Parameters}\label{tab:simulation-parameters}
\begin{tabular}{{|l|p{140pt}|}}
\hline
Network topology                    & NSFNET\\
\hline
$\mathcal {D}$                      & \{3, 5, 8, 10, 12\}             \\
\hline
$\varsigma$                         & $10$ VMs                          \\
\hline
$c_{m}$                             & $1000$ servers                     \\
\hline
$\alpha=\{\alpha_{1},\alpha_{2},...,\alpha_{m}\}$& $\{2.1, 2.5, 1.9, 2.8, 2\}$        \\
\hline
$\varTheta_{m}$                     & $[0, 8000]$\\
\hline
$\varPhi_{m}$                       & $[1000, 9000]$ \\
\hline
$p_{s}$                             & $10$ units, $1$ unit for $1$ VM in average  \\
\hline
$\zeta_{m,k}$                       & $[1, 14]$ Gb/s                            \\
\hline
$c_{e}$                             & $300$ spectrum slots                       \\
\hline
$c_{s}$                             & $12.5$ Gbps                                  \\
\hline
$\kappa$                            & $\{2, 4, 8, 16\} $ spectrum slots              \\
\hline
$\frac{\lambda}{\mu}$               & $\{40, 80, 120, 160, 200, 240, 280, 320\}$      \\
\hline
\end{tabular}
\end{center}
\end{table}

We run the simulation for 150 times, and exclude the scenario with empty VM requests traffic load $(\mathcal{D}_{s}\neq\varnothing \quad\&\quad \mathcal{D}_{d}\neq\varnothing)$. Fig. \ref{fig:compare_energy} shows the total cost of brown energy consumption of the strategy without using renewable energy, Manycast-SPR $(\kappa =2)$ and  Manycast-LPR $(\kappa = 2)$. Apparently, Manycast-SPR and Manycast-LPR can save brown energy substantially. Manycast-SPR saves about $15\%$ cost of brown energy as compared with the strategy without migration. Manycast-LPR reduces up to $31\%$ cost of brown energy as compared with the strategy without migration. Manycast-LPR has better performance because Manycast-LPR employs the least weight path $p$ of all node pairs for routing, while Manycast-LPR engages only the short path $p$ of one node pair.

\begin{figure}[!htb]
    \centering
    \includegraphics[width=1.0\columnwidth]{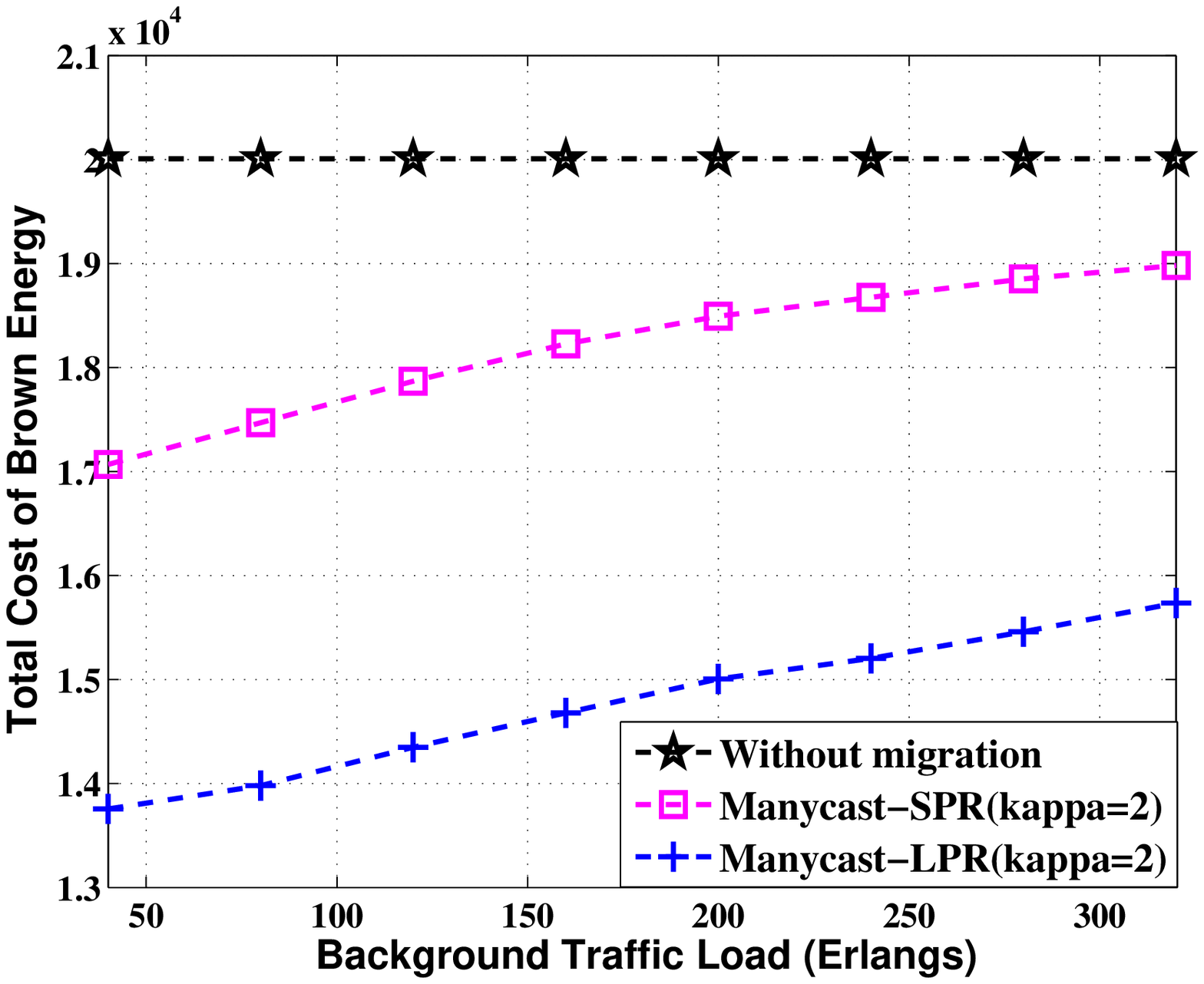}
    \caption{\small Total brown energy cost comparison.}
    \label{fig:compare_energy}
\end{figure}
\begin{figure}[!htb]
    \centering
    \includegraphics[width=1.0\columnwidth]{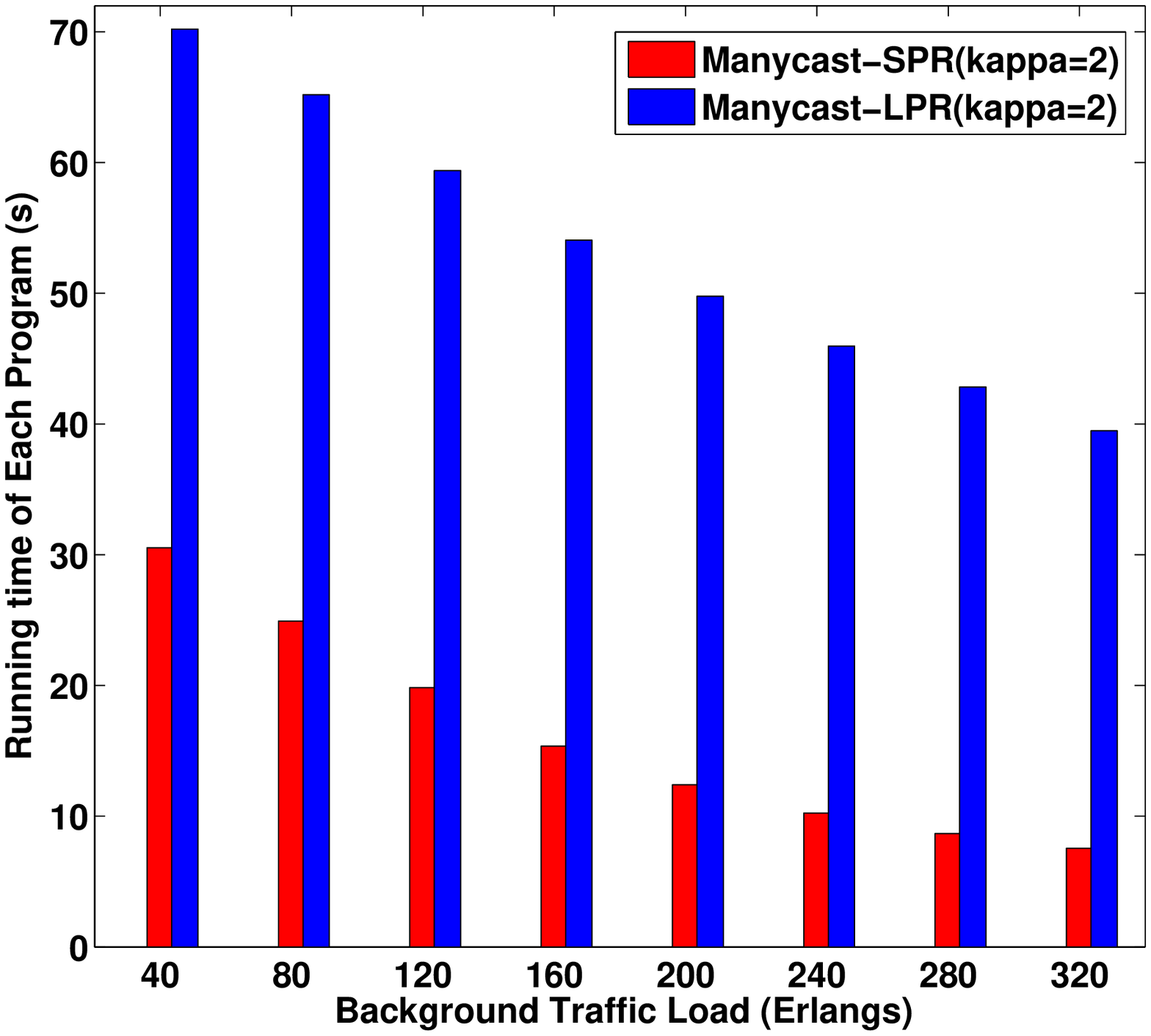}
    \caption{\small Running time comparison.}
    \label{fig:compare_time}
\end{figure}

In order to obtain a better analysis, the running time of Manycast-SPR $(\kappa = 2)$ and Manycast-LPR $(\kappa = 2)$ are shown in Fig. \ref{fig:compare_time}. Manycast-SPR spends less time than Manycast-LPR, implying that Manycast-SPR has a lower complexity and Manycast-LPR has a higher computing complexity. It also illustrates that the time and the final cost value is a trade-off in the evaluation. Manycast-LPR is more complex and hence incurs a lower brown energy cost.

\begin{figure}[!htb]
    \centering
    \includegraphics[width=1.0\columnwidth]{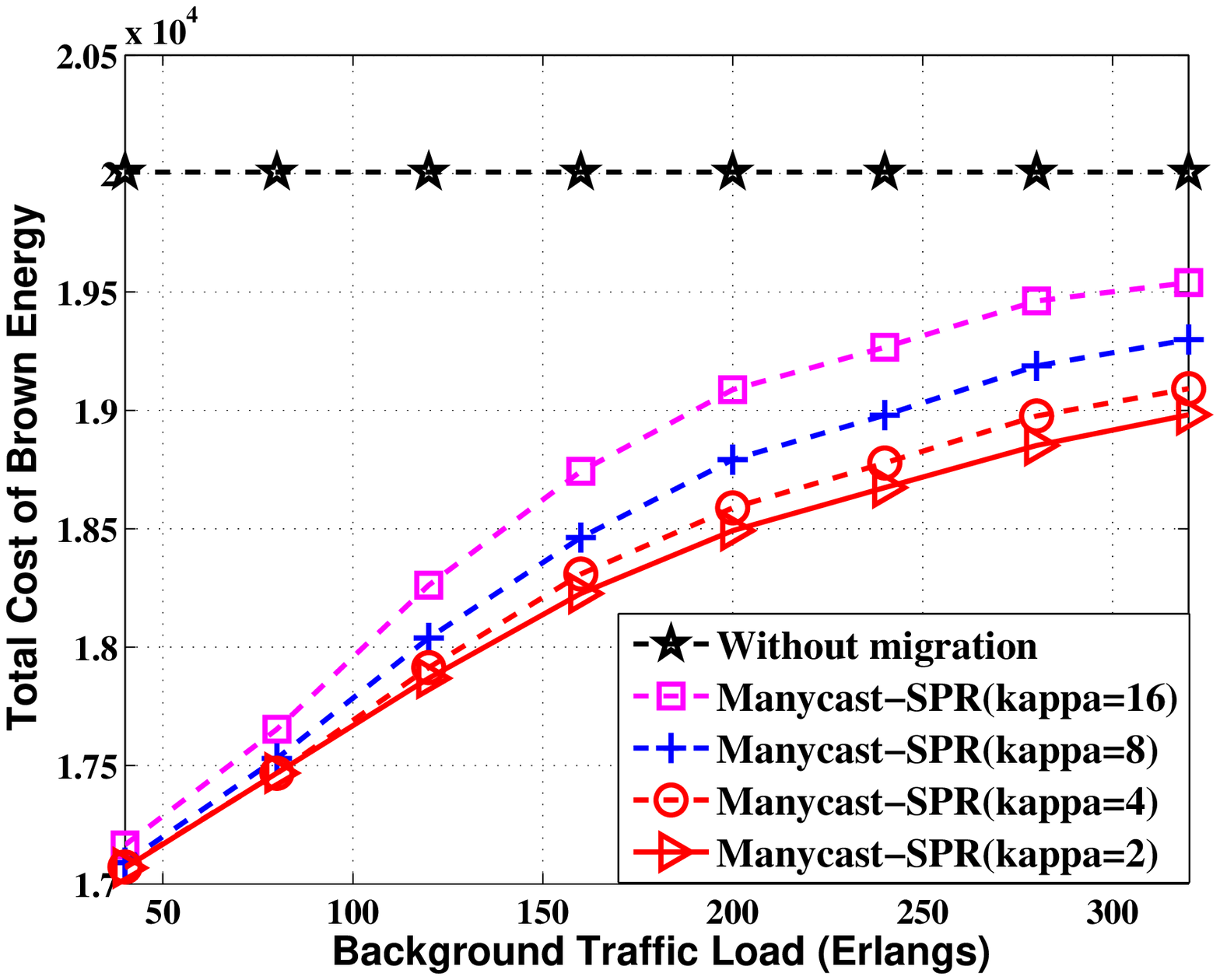}
    \caption{\small Total brown energy cost of Manycast-SPR.}
    \label{fig:cm1000ks1}
\end{figure}
\begin{figure}[!htb]
    \centering
    \includegraphics[width=1.0\columnwidth]{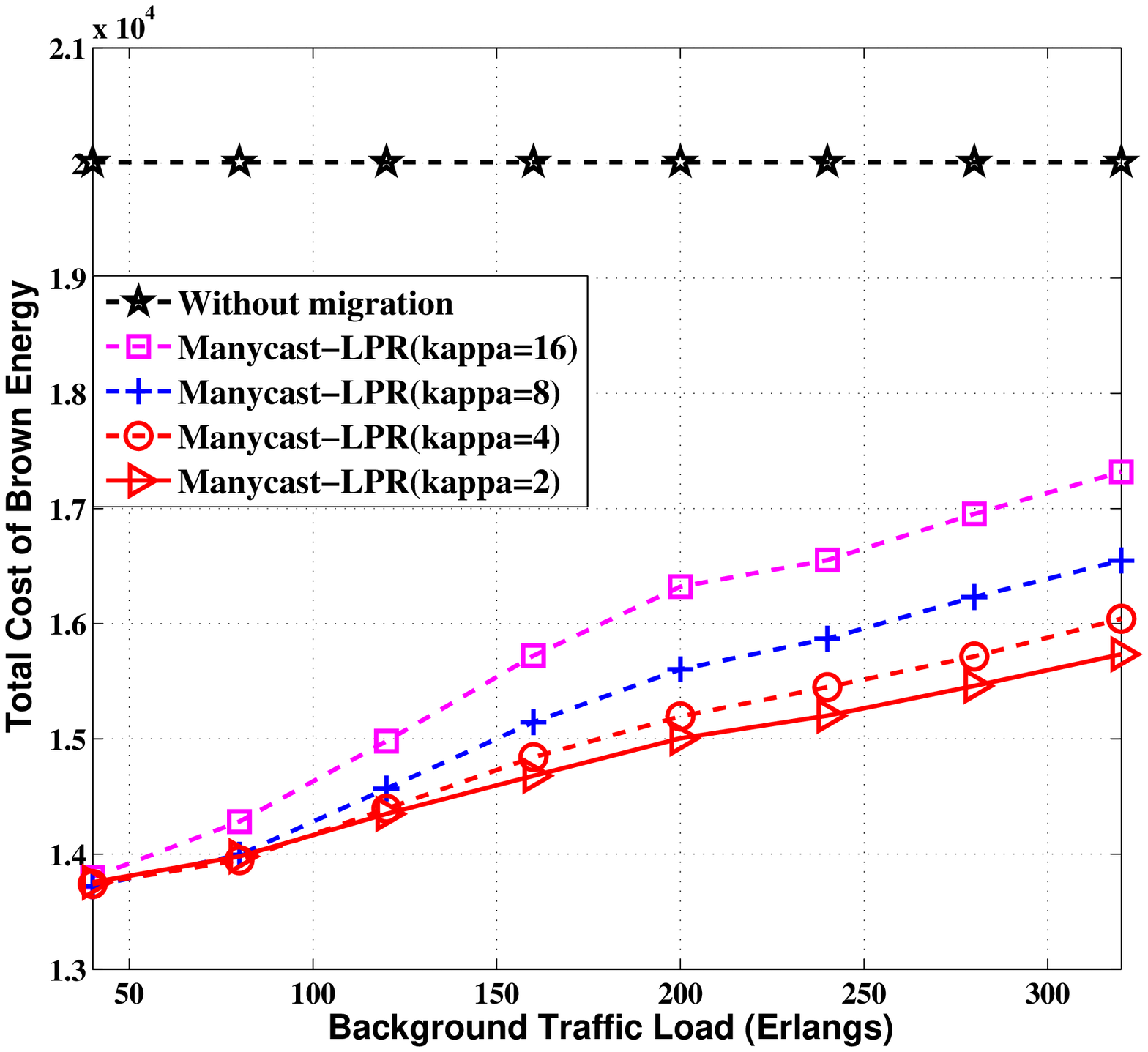}
    \caption{\small Total brown energy cost of Manycast-LPR.}
    \label{fig:cm1000ks2}
\end{figure}

The results of Manycast-SPR for various $\kappa$ are described in Fig. \ref{fig:cm1000ks1}. The cost of brown energy consumption keeps increasing when the background traffic increases, because high background traffic tends to congest the network links and leads to more migration failures. Apparently, a small $\kappa$ brings more benefits than a big $\kappa$ in reducing the energy cost.

Fig. \ref{fig:cm1000ks2} shows the results of Manycast-LPR for various $\kappa$, almost the same results as shown in Fig. \ref{fig:cm1000ks1}, but the cost of the brown energy consumption is much less than that in Fig. \ref{fig:cm1000ks1}, because Manycast-LPR can easily find a path which has available bandwidth for migration. Obviously, Manycast-LPR with $\kappa =2$ achieves the best result with the lowest cost of consumed brown energy. All these results illustrate that a small $\kappa$ leads to a lower cost of the brown energy consumption and a big $\kappa$ induces a higher cost of the brown energy consumption. This is because it is difficult to find a path with enough bandwidth for a big $\kappa$, when the network has background traffic. A smaller $\kappa$ achieves a lower energy cost at the cost of higher complexity.
%A small $\kappa$ gives rise to higher complexity and consumes more time but produces results with a low %energy cost.

Figs. \ref{fig:SPR_time} and \ref{fig:LPR_time} show the running time of Manycast-SPR and that of Manycast-LPR with different $\kappa$, respectively. We can observe that the computing time is decreased when the traffic load increases. For the same $\kappa$ with a given background traffic load, Manycast-SPR consumes more time than Manycast-LPR does. For either of the two algorithms under a specific background traffic load, we can see that the running time is nearly halved when $\kappa$ is doubled. Hence, a smaller $\kappa$ brings a better performance but takes longer time, and a larger $\kappa$ has worse performance with a shorter running time.

%\textcolor{blue}{From our results, Manycast-SPR does not have better performance than Manycast-LPR, it is because Manycast-LPR consuming more time and it returns with better results. In our simulation, the generation of the renewable energy in many DCs do not match the computing demands; in other words, there are many inter-DC VM migration requests. If both the network background traffic and migration traffic are low, the Manycast-SPR and Manycast-LPR should have the same performance, but Manycast-SPR expands less time than Manycast-LPR.}

\begin{figure}[!htb]
    \centering
    \includegraphics[width=1.0\columnwidth]{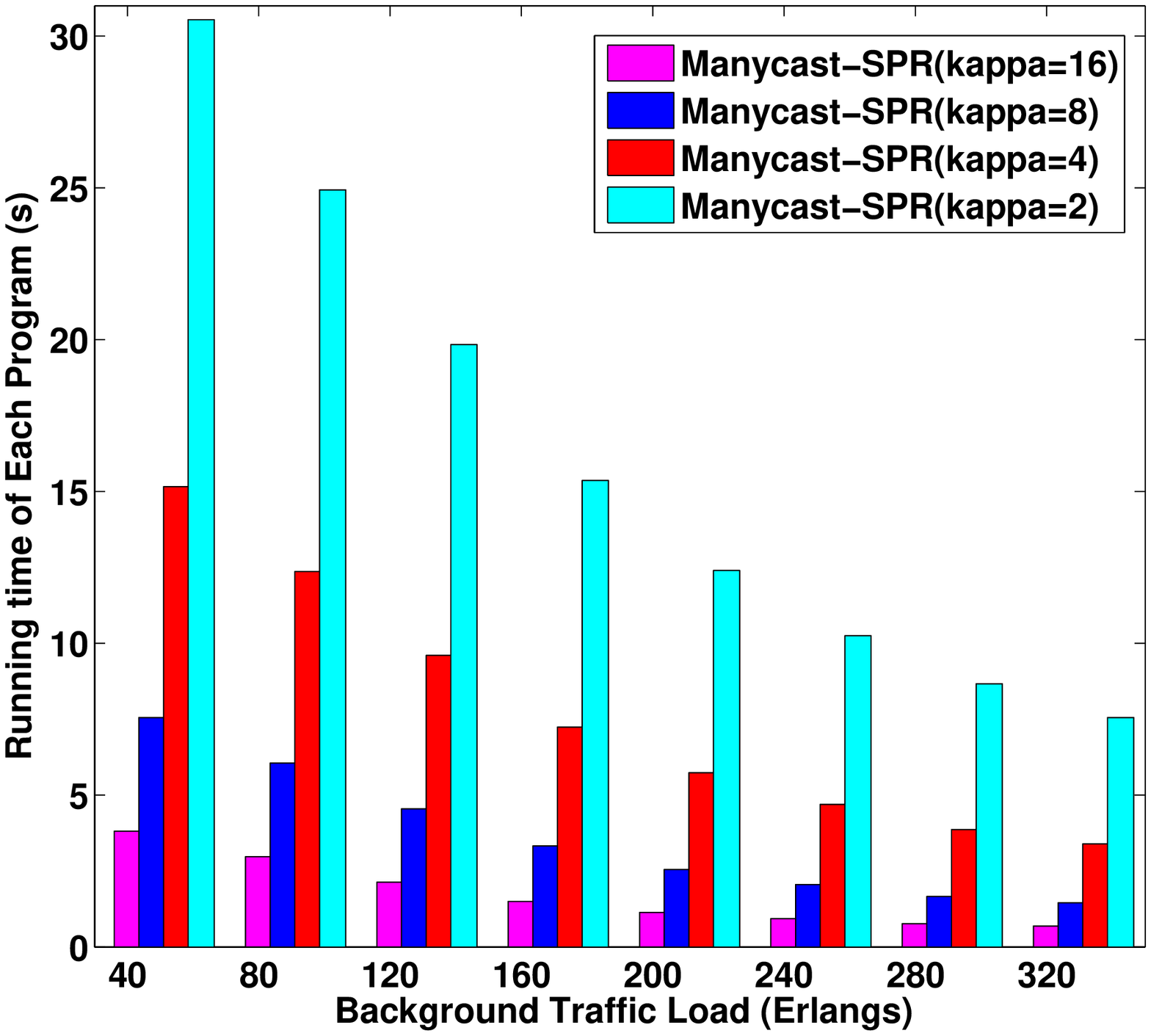}
    \caption{\small Running time of Manycast-SPR.}
    \label{fig:SPR_time}
\end{figure}
\begin{figure}[!htb]
    \centering
    \includegraphics[width=1.0\columnwidth]{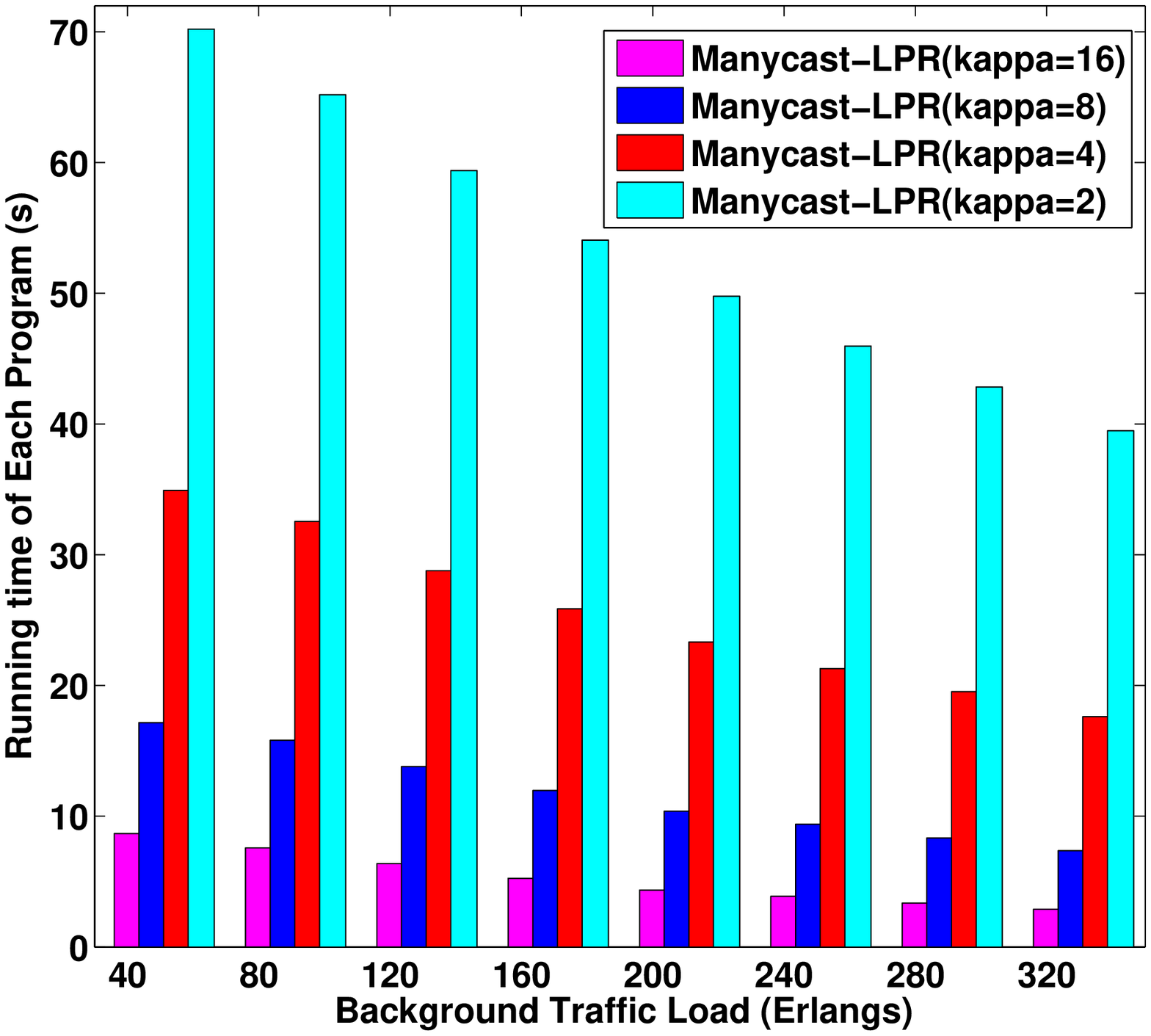}
    \caption{\small Running time of Manycast-LPR.}
    \label{fig:LPR_time}
\end{figure}

\section{Conclusion}\label{conclusion}
Datacenters are widely deployed for the increasing demands of data processing and cloud computing. The energy consumption of DCs will take up $25\%$ of the total ICT energy consumption by 2020. Powering DCs with renewable energy can help save brown energy. However, the availability of renewable energy varies by locations and changes over time, and DCs' work loads demands also vary by locations and time, thus leading to the mismatch between the renewable energy supplies and the work loads demands in DCs. Inter-DC VM migration brings additional traffic to the network, and the VM mitigation is constrained by the network capacity, rendering inter-DC VM migration a great challenge.

This paper addresses the emerging renewable energy-aware inter-DC VM migration problem. The main contribution of this paper is to reduce the network influence on the migration while minimizing the brown energy consumption of the DCs. The RE-AIM problem is formulated and proven to be NP-hard. Two heuristic algorithms, Manycast-SPR and Manycast-LPR, have been proposed to solve the RE-AIM problem. Our results show that Manycast-SPR saves about $15\%$ cost of brown energy as compared with the strategy without migration, while Manycast-LPR saves about $31\%$ cost of brown energy as compared with the strategy without migration. The computing time of Manycast-LPR is longer than that of Manycast-SPR because the complexity of Manycast-LPR is higher than Manycast-SPR. In conclusion, we have demonstrated the viability of the proposed algorithms in minimizing brown energy consumption in inter-DC migration without congesting the network.

%\section*{Acknowledgments}

\bibliographystyle{IEEEtran}

%\bibliography{IEEEabrv,mybibfile}
\end{document}